\documentclass[conference]{IEEEtran}
\IEEEoverridecommandlockouts
\usepackage[noadjust]{cite}
\usepackage[colorlinks,linkcolor=blue,citecolor=blue]{hyperref}
\usepackage{amsmath,amssymb,amsfonts}
\usepackage{algorithmic}
\usepackage{graphicx}
\usepackage{subfig}
\usepackage{textcomp}
\usepackage{xcolor}
\usepackage[capitalize,nameinlink]{cleveref}
\usepackage{mathrsfs}
\usepackage{amsthm}
\theoremstyle{remark}
\newtheorem{Remark}{Remark}

\newtheorem{Lemma}{Lemma}
\makeatletter
\renewenvironment{proof}[1][\proofname]{\par
\pushQED{\qed}%
\normalfont \topsep0\p@\@plus0\p@\relax
\trivlist
\item\relax
{\itshape
#1\@addpunct{.}}\hspace\labelsep\ignorespaces
}{%
\popQED\endtrivlist\@endpefalse
\vspace{0.5\baselineskip}
}
\makeatother

\usepackage{etoolbox}
\pretocmd{\appendices}{\crefalias{section}{appendix}}{}{}
\crefname{figure}{Fig.}{Figs.}
\crefname{equation}{}{}
\crefname{Lemma}{Lemma}{Lemmas}

\def\BibTeX{{\rm B\kern-.05em{\sc i\kern-.025em b}\kern-.08em
    T\kern-.1667em\lower.7ex\hbox{E}\kern-.125emX}}
    
\begin{document}

\title{WiSLAT: A Simultaneous Device Localization and Target Tracking Method for Wi-Fi Systems}

\author{
    \IEEEauthorblockN{Chunxi Chen\IEEEauthorrefmark{1}, Jingwen Zhang\IEEEauthorrefmark{1}, Chao Yu\IEEEauthorrefmark{2}, Fan Liu\IEEEauthorrefmark{3}, Rui Wang\IEEEauthorrefmark{2}}
    \IEEEauthorblockA{\IEEEauthorrefmark{1} College of Semiconductors, Southern University of Science and Technology, Shenzhen, China}
    \IEEEauthorblockA{\IEEEauthorrefmark{2} Department of Electronic and Electrical Engineering, Southern University of Science and Technology, Shenzhen, China}
    \IEEEauthorblockA{\IEEEauthorrefmark{3} School of Information Science and Engineering, Southeast University, Nanjing, China}
    \thanks{This work was supported in part by the National Science and Technology Major Projects of China under Grant 2025ZD1302000, and the National Natural Science Foundation of China (NSFC) under Grant 62522107.}
}

\maketitle
\IEEEpeerreviewmaketitle
\begin{abstract}
It has been shown that the channel state information (CSI) of a Wi-Fi system can be exploited to localize Wi-Fi devices or track trajectory of a moving target. In the existing literature, both sensing tasks are treated separately and some prior information is usually requested, including the signal fingerprints, the locations of some anchor devices in the Wi-Fi system, and etc. In the proposed WiSLAT method, however, it is shown that both sensing tasks can assist each other, such that the request on prior system information can be eliminated. Particularly, in a Wi-Fi system with an access point (AP) and at least three stations, where the locations of the stations are unknown, the WiSLAT is designed to detect the Doppler frequencies of the downlink CSI at the stations, such that their locations and the trajectory of the target with respect to the AP can be inferred. The joint detection can be conducted by searching the optimal stations' locations and target's trajectory, such that their corresponding Doppler frequencies fit the observed ones best. Due to the tremendous non-convex search space, a low-complexity sub-optimal algorithm integrating alternate optimization, extended Kalman filter and density-based clustering is proposed in WiSLAT. Experiments conducted in indoor environments demonstrate the effectiveness of WiSLAT, achieving a median trajectory-tracking error of 0.68 m.
\end{abstract}

% === reduce the line spacing ===
\setlength{\belowdisplayskip}{0pt}
\setlength{\abovedisplayskip}{0pt}
\section{Introduction}
The Wi-Fi signal has been widely exploited in a number of sensing scenarios, including motion recognition, passive target tracking, device localization and etc. In most of the existing literature, the above sensing scenarios are addressed in an isolated manner with dedicated prior information. The motion recognition and passive target tracking usually rely on multiple Wi-Fi stations, whose locations are precisely known in advance. The Wi-Fi device localization usually requests a measurement of signal fingerprint in advance, which also raise a localization request, or prior location knowledge of multiple coordinated stations. Hence, it is interesting to investigate the Wi-Fi sensing technique at a cold start, where the requirement on the system's prior knowledge is minimized or eliminated.

There have been a number of works investigating the tracking technique for a passive target via Wi-Fi signal. These works can be generally classified into two distinct paradigms: data-driven approaches and model-driven approaches. Data-driven methods mainly exploit the characters or statistics of Channel State Information (CSI) to determine the location of a passive target \cite{monophy}, \cite{pilot}. Note that the CSI extracted from commercial off-the-shelf (COTS) Wi-Fi devices often exhibits significant noise, particularly phase noise. A number of research efforts have been focused on denoising techniques to improve signal quality. For example, CARM \cite{carm} and Widar \cite{widar} systems applied Principal Component Analysis (PCA) to leverage signal correlations across different subcarriers and remove noise, while preserving signal components introduced by motion. 
IndoTrack \cite{indotrack} and Widar2.0 \cite{widar2} leveraged the fact that random phase offsets are identical among different antennas on the same Network Interface Card (NIC), and employed conjugate multiplication to eliminate the phase noise. 
Based on the same fact, FarSense \cite{farsense} and WiTraj \cite{witraj} proposed the CSI-ratio or CSI-quotient methods to resolve the ambiguity issue due to the multiplication terms between dynamic and static components in CSI.
More recent works \cite{rtcsi} and ML-track \cite{mltrack} used  Round-Trip CSI (RTCSI) and Multi-Link RTCSI (ML-RTCSI) to eliminate the effect of Carrier Frequency Offset (CFO), using the fact that CFO of the uplink and downlink are opposite.

On the other hand, model-based methods utilized the Time-of-Flight (ToF) \cite{multitrack}, Angle-of-Arrival (AoA) \cite{matrack}, Doppler Frequency Shift (DFS) \cite{widar,witraj,rtcsi,mltrack} or their combinations \cite{indotrack,widar2,widfs} in the localization or tracking of a passive target. The ToF and AoA data extracted from commercial off-the-shelf (COTS) Wi-Fi devices often suffer from limited antenna number and low resolution. Hence in \cite{MultiAoA}, it requires multiple receive stations to jointly reduce measurement errors. The DFS-based methods are usually more robust against the detection error. For example, WiPIHT\cite{WiPIHT} eliminates reliance on fixed initial user locations by extracting location-independent velocity features and decoupling motion direction. 
However, all these localization or tracking methods for a passive target rely on more than one communication devices, where their relative locations are known in advance. Hence, it is still unknown that how to achieve passive target tracking without such prerequisites on devices' location information.

There have also been significant research efforts devoted to the localization of Wi-Fi devices, which can be categorized into two classes: fingerprint-based localization and geometry-based localization. 
The fingerprint-based localization relies on extensive prior measurements of a database mapping locations to Wi-Fi signal signatures \cite{fingerprint3D,CRISLoc}. 
Geometry-based localization methods leverage the estimation of the physical or geometric parameters of signal propagation to directly infer the locations of Wi-Fi devices. 
For example, in ToneTrack\cite{ToneTrack} fused multi-Tone Time of Arrival (ToA) to overcome the temporal resolution barrier posed in narrowband channel.
In \cite{tdoa}, the Time Difference of Arrival (TDoA) method estimate the distance by measuring the difference of propagation time of radio signals between devices, requiring strict clock synchronization between the transmitter and the receiver. In \cite{TriangAoA,MultiAoA,SpotFi}, methods based on AoA estimate the direction of a received signal by measuring the phase or delay differences across an antenna array at the receiver, then combine multiple observations to determine the location. In the above works on device localization, the prior information on the devices' relative locations or fingerprints is necessary. In summary, prior system or environment information is necessary in localization or tracking of either a passive target or Wi-Fi devices in the existing literature.

In this paper, we would like to shed some lights on the above open problem by proposing a method of simultaneous localization and tracking (SLAT) for Wi-Fi systems, namely WiSLAT. The WiSLAT is designed to track the trajectory of one single moving target without prior location knowledge on the sensing stations (STAs) relative to the AP. Instead, these unknown locations as well as the unknown trajectory of one moving target can be estimated from the Doppler detections of the stations jointly, as long as the number of detecting stations is not less than 3. Particularly in WiSLAT, the joint detection of the stations' locations and the moving target's trajectory is formulated as a minimum mean square error (MMSE) problem, and a low-complexity solution algorithm is proposed. In summary, the main contributions of this work are as follows:
\begin{enumerate}
\item To the best of our knowledge, this is the first work investigating the joint estimation of Wi-Fi stations' locations and a moving target's trajectory. Both estimations actually facilitate each other in the proposed WiSLAT. Thus, it is infeasible to estimate the Wi-Fi stations' locations without the moving target, and vice versa. 
\item In order to solve the above joint estimation problem, we propose a low-complexity algorithm to address the issue of significant computation complexity. It consists of a coarse search stage, yielding a good initial solution, and a iterative fine adjustment stage, refining the locations with gradient-based search.
\item The experiments are conducted to verify the performance of the WiSLAT method. It is demonstrated that the median error of trajectory tracking and devices localization are 0.68m and 1.07m, respectively.  Moreover, by averaging the results of multiple trajectories, the final error of devices localization is reduced to 0.45m on average.
\end{enumerate}

The remainder of the paper is organized as follows. Section II presents the motion model and the signal model, where the proposed WiSLAT method will be applied. Section III formulates the  joint receive stations' localization and target's tracking design as a MMSE problem. 
Section IV proposes a low-complexity optimization algorithm to solve the above MMSE problem. The experimental setup, results, and in-depth discussions are provided in Section V. Finally, the conclusion is drawn in Section VI.

\section{System Model}

The WiSLAT method is designed for Wi-Fi systems to track the trajectory of a moving target with respect to the access point (AP), and meanwhile, detect the locations of the Wi-Fi stations, which participate in the trajectory tracking. Thus, it is not necessary to know the their locations in advance. The trajectory tracking and location detection are jointly designed. It is not a trivial aggregation of two tasks. In fact, any of the two tasks cannot be achieved alone. We shall refer to the above joint design as Simultaneous device Localization and trajectory Tracking (SLAT). In this section, we first establish the motion model connecting the target's trajectory (including the initial position and time-varying velocity) with the Doppler measurements, then elaborate the method adopted by WiSLAT in Doppler detection.

\subsection{Motion Model}
\begin{figure}[!h]
    \centering
    \includegraphics[width=0.9\columnwidth]{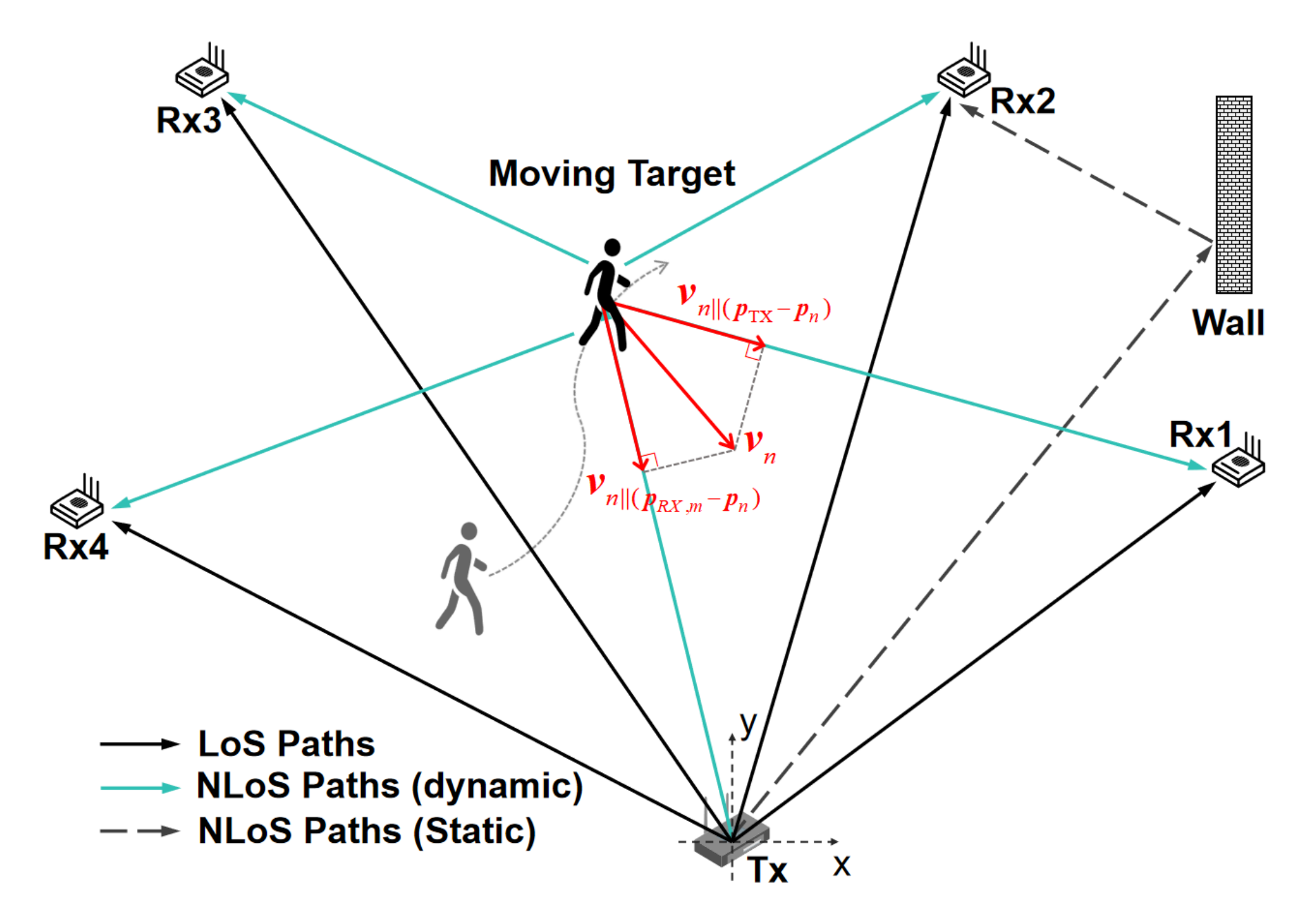}
    \caption{An scenario illustration of simultaneous devices localization and trajectory tracking (SLAT) with one AP and four receive stations.}
    \label{fig:overview}
\end{figure}
The proposed WiSLAT is designed for a Wi-Fi system consisting of one AP and at least three stations. The stations receive downlink signals from the AP, such that the downlink channel state information (CSI) can be estimated. Denote the $m$-th receive station as
Rx-$m\,(m=1,...,M)$. A single target moves within the coverage area, such that Doppler effect is raised at the downlink CSI of the receive stations, as shown in the \cref{fig:overview}.

In the two-dimensional plane of the detection, let $\boldsymbol{p}_\mathrm{TX}=[0,0]^\mathsf{T}$ and  $\boldsymbol{p}_{\mathrm{RX},m}=[x_{\mathrm{RX},m},y_{\mathrm{RX},m}]^\mathsf{T}$ ($m=1,...,M$) denote the position of AP and Rx-$m$, respectively. 
The target's continuous trajectory is uniformly sampled with a time interval $\Delta t$.
 Let $\boldsymbol{p}_n=[x_n,y_n]^\mathsf{T}$ be the position of the target at the $n$-th time instance (the beginning of the $n$-th time interval), the trajectory of the target can be represented by the matrix 
 \begin{equation}
    \mathbf{J}=[\boldsymbol{p}_1,\boldsymbol{p}_2,...,\boldsymbol{p}_N],
 \end{equation} 
 where it is assumed that the target moves for $N-1$ time intervals.

With a sufficiently small $\Delta t$, the target's velocity within a number of successive time intervals can be treated as a constant. Let $\boldsymbol{v}_n=[v_{x_n},v_{y_n}]^\mathsf{T}$ be the velocity of the target in the $n$-th time interval, the motion of the target can be represented as 
\begin{equation}
    \boldsymbol{p}_n = \boldsymbol{p}_{n-1}+\boldsymbol{v}_{n-1}\Delta t, \ n=2,3,...,N.
    \label{con:pn}
\end{equation}
Hence, the trajectory of the target can be equivalently represented by the matrix
\begin{equation}
    \mathbf{J}=[\boldsymbol{p}_1,\boldsymbol{v}_1,...,\boldsymbol{v}_{N-1}].\label{con:Jpv}
\end{equation}

As a result, the Doppler frequency at the Rx-$m$ in the $n$-th time interval ($n=1,2,...,N-1$) raised by the moving target  can be written as 
\begin{equation}
    f^d_{m}(n) = \frac{1}{\lambda}\left(\frac{\boldsymbol{p}_\mathrm{TX} - \boldsymbol{p}_n}{|\boldsymbol{p}_\mathrm{TX} - \boldsymbol{p}_n|} + \frac{\boldsymbol{p}_{\mathrm{RX},m} -\boldsymbol{p}_n}{|\boldsymbol{p}_{\mathrm{RX},m} - \boldsymbol{p}_n|}\right)^\mathsf{T} \boldsymbol{v}_n  ,
    \label{con:fd}
\end{equation}
where $\lambda$ denotes the carrier wavelength\footnote{In fact, the frequencies of subcarriers are different in a OFDM signal, leading to slightly different Doppler shifts. In this paper, such difference can be neglected in the localization and trajectory tracking. Hence, we use the central carrier's wavelength as the value of $\lambda$.}. The Doppler frequency \cref{con:fd} is given by the sum of the projections of  $\boldsymbol{v}_n$ onto the vectors $(\boldsymbol{p}_\mathrm{TX} - \boldsymbol{p}_n)$ and $(\boldsymbol{p}_{\mathrm{RX},m} - \boldsymbol{p}_n)$, respectively, as illustrated as in \cref{fig:overview}.

\subsection{Doppler Detection via CSI}
\begin{figure}[!h]
    \centering
    \includegraphics[width=0.99\columnwidth]{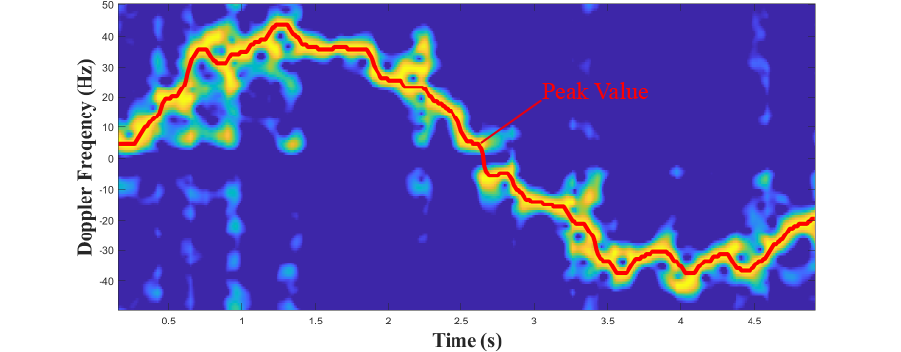}
    \caption{An example of time-frequency spectrogram in the Doppler detection.}
    \label{fig:doppler}
\end{figure}
The Doppler frequencies of the moving target at all the time intervals are detected in WiSLAT. Without loss of generality, the Doppler detection for the $n$-th time interval is elaborated in this part. In the detection, it is assumed that the duration of time interval $\Delta t$ is sufficiently small, such that the Doppler frequency raised by the moving target is quasi-static within $2Q+1$ time intervals. Hence, the Doppler frequency of the $n$-th time interval can be approximately obtained from the Downlink CSI, which is estimated once in each time interval, from the $(n-Q)$-th time interval to the $(n+Q)$-th one.

Let $H_{p,m}(q)$ be the CSI sample of one subcarrier from one antenna of the AP to the $p$-th antenna ($p=1,2$) of Rx-$m$ in the $q$-th time interval ($q=n-Q,...,n+Q$). It can be expressed as 
\begin{equation}
    \begin{aligned}
    &{H}_{p,m}(q) =\sum^{L}_{l=0}{A}_{p,m,l} e^{-j2\pi\left [f^{d}_{p,m,l}(q-n+Q)\Delta t+{\omega}_m(q)\right]}\\
    & =e^{-j2\pi{\omega}_m(q)} \left[ {H}_{p,m,0}+ {A}_{p,m,1} e^{-j2\pi f^{d}_{p,m,1}(q-n+Q)\Delta t}\right],
    \end{aligned}\label{con:channel}
\end{equation}
where $L$ is the number of NLoS propagation paths, the $0$-th and $1$-st paths refer to the LoS path and the NLoS path scattered off the moving target, ${A}_{p,m,l}$ and $f^{d}_{p,m,l}$ denote the complex path gain and Doppler frequency of the $l$-th path respectively, $\omega_m(q)$ represents the time-varying phase offset due to asynchronous hardware between the transmitter and the receiver. Assuming the Doppler 
frequency except the first path is zero,
\begin{equation}
    {H}_{p,m,0} \approx \sum_{l \neq 1}{A}_{p,m,l}
\end{equation}
is quasi-static in the Doppler detection.

Because the existence of the time-varying phase offset $\omega_m(p)$, the Doppler raised by the moving target cannot be directly detected from the CSI of single antenna. Note that all antennas on the same Wi-Fi adapter share a common RF oscillator, they have the same hardware phase offset $\omega_m(p)$ as well. Hence, the method of cross-antenna signal ratio proposed in \cite{farsense} is applied to eliminate the hardware phase offset and detect the desired Doppler 
frequencies. Since the two antennas of each receive station are usually close, we adopt the approximation  $f^d_{2,m,1} \approx f^d_{1,m,1}$ in the detection window. Hence, the ratio of two antenna's CSI in the $q$-th time interval ($q=n-Q,...,n+Q$) can be derived as 
\begin{equation}
\begin{aligned}
    {R}_m(q)&  = \frac{{H}_{1,m}(q)}{{H}_{2,m}(q)}\\
    & = \frac{{H}_{1,m,0}+ {A}_{1,m,1} e^{-j2\pi f^{d}_{1,m,1}(q-n+Q)\Delta t}} {{H}_{2,m,0}+ {A}_{2,m,1} e^{-j2\pi f^{d}_{1,m,1}(q-n+Q)\Delta t}}\\
    & = \frac{{H}_{1,m,0}+ {A}_{1,m,1} z_q} {{H}_{2,m,0}+ {A}_{2,m,1} z_q}\\
    &= \frac{{A}_{2,m,1}{H}_{1,m,0} - {A}_{1,m,1}{H}_{2,m,0} }{{A}_{2,m,1}^2}\\
    &\cdot \frac{1}{z_q + \frac{{H}_{2,m,0}}{{A}_{2,m,1}}} + \frac{{A}_{1,m,1}}{{A}_{2,m,1}},
\end{aligned}\label{con:ratio}
\end{equation}
where the time-varying hardware phase offset of the two antennas is cancelled in the second equality,
\begin{equation}
    z_q \triangleq e^{-j2\pi f^{d}_{1,m,1}(q-n+Q)\Delta t}
\end{equation}
rotates with $q$ in a frequency of $f^{d}_{1,m,1}$, the fourth equality is due to the Mobius transformation \cite{Moebius}.

As discussed in \cite{farsense}, due to the existence of static LoS path, the ratio ${R}_m(q)$ and $z_q$ are in phase. They share the same frequency as $f^{d}_{1,m,1}$, which is the desired Doppler frequency $f_m^d(n)$ as defined in \cref{con:fd}. Hence, the Short-Time Fourier Transform (STFT) is applied to detect the frequency components in ${R}_m(q)$. 

Particularly, we define the time-frequency matrix $\mathbf{P}_m(\xi,n)$ as
\begin{equation}
    \mathbf{P}_m(\xi,n) =
    \sum^{Q}_{i=-Q} {R}_m(n+i)e^{-j2\pi\frac{i+Q}{N_\text{FFT}}\xi},
\end{equation}
where $N_\text{FFT}$ is the FFT size. 
An example of $\mathbf{P}_m(\xi,n)$ is illustrated in \cref{fig:doppler}. It can be observed that in addition to the dominant Doppler component of the moving target, multiple weaker, closely-spaced Doppler components could also exist. This is due to the presence of multiple scattering points on moving target. Therefore, we choose the Doppler component with the maximum magnitude as the detected Doppler shift of the moving target. Thus, the detected Doppler shift of the $n$-th time interval is given by
\begin{equation}
    \tilde{f^d_m}(n) =\begin{cases}
    \frac{\xi^*}{N_\text{FFT}}f_s,& 0\leq\xi^*<\frac{N_\text{FFT}}{2},\\
    \frac{\xi^*}{N_\text{FFT}}f_s-f_s,& \frac{N_\text{FFT}}{2}\leq\xi^*<{N_\text{FFT}},
    \end{cases}
\end{equation}
where
\begin{equation} 
\xi^* = \arg\max_{\xi}|\mathbf{P}_m(\xi,n)| \ \text{and}\ f_s = \frac{1}{\Delta t}.
\end{equation}

\section{WiSLAT Problem Formulation}
The WiSLAT method is proposed to localize the fixed receive stations and the moving target jointly. Note that their location information can not be directly inferred from the Doppler measurement. In this section, we shall formulate the joint localization and tracking design as an indirect matching problem: find the best locations of the receive stations and the trajectory of the target, such that the corresponding time-varying Doppler 
frequencies fit the real measurements with the minimum Mean Square Error (MSE).

Denote the aggregation of the $M$ receive stations' locations as
\begin{equation}
    \mathbf{R}=[\boldsymbol{p}_{\mathrm{RX},1},...,\boldsymbol{p}_{\mathrm{RX},M}].
\end{equation}
Given the trajectory and locations $(\mathbf{J},\mathbf{R})$, the aggregation of the corresponding Doppler frequencies at the M receive stations can be expressed as
\begin{equation}
    \mathbf{Z}(\mathbf{J},\mathbf{R}) = 
    \begin{bmatrix}
f^d_1(1) & \cdots & f^d_1(N-1)\\
\vdots & \cdots & \vdots\\
f^d_M(1) & \cdots & f^d_M(N-1)\\
\end{bmatrix}  ,
\end{equation}
where $f^d_m(n)$ can be calculated by \cref{con:fd}. On the other hand, the matrix aggregating all the measured Doppler frequencies at all the receive stations in all the time interval can be written as
\begin{equation}
    \widetilde{\mathbf{Z}} =  \begin{bmatrix}
\tilde{f^d_1}(1) & \cdots & \tilde{f^d_1}(N-1)\\
\vdots & \cdots & \vdots\\
\tilde{f^d_M}(1) & \cdots & \tilde{f^d_M}(N-1)\\
\end{bmatrix}.
\end{equation}
For the notation convenience, the $m$-th row vectors of $\mathbf{Z}(\mathbf{J},\mathbf{R})$ and $\widetilde{\mathbf{Z}}$ are denoted as $\boldsymbol{z}_m(\mathbf{J}^{(\kappa-1)},\boldsymbol{p}_{\mathrm{RX},m})$ and $\widetilde{\boldsymbol{z}}_m$, respectively.

Accordingly, the MSE between $\mathbf{Z}$ and $\widetilde{\mathbf{Z}}$ is defined as
\begin{equation}
    g(\mathbf{J},\mathbf{R},\widetilde{\mathbf{Z}}) = \frac{1}{M(N-1)}\| \widetilde{\mathbf{Z}} -  \mathbf{Z}(\mathbf{J},\mathbf{R})\|_\mathrm{F}^2 ,
    \label{con:g}
\end{equation}
where $\|\cdot\|_\mathrm{F}$ denotes the Frobenius Norm. As a result, the joint search of the receive stations' locations and the target's trajectory can be formulated as the following Minimum MSE (MMSE) problem:
\begin{equation}
    \mathcal{P}_{0} : \quad \left\{\mathbf{J}^*,\mathbf{R}^*\right\}=\arg \min_{\left\{\mathbf{J},\mathbf{R}\right\}}g(\mathbf{J},\mathbf{R},\widetilde{\mathbf{Z}}) .
\end{equation}

The variables of the above optimization problem consists of the locations of the $M$ receive stations, the starting point and velocities of the trajectory. Exhaustive search of them will lead to prohibitive computation complexity. In the following section, a low-complexity search algorithm will be proposed. 

\section{Low-Complexity Solution Algorithm}
\begin{figure*}[ht]
    \centering

    \includegraphics[width=0.9\textwidth]{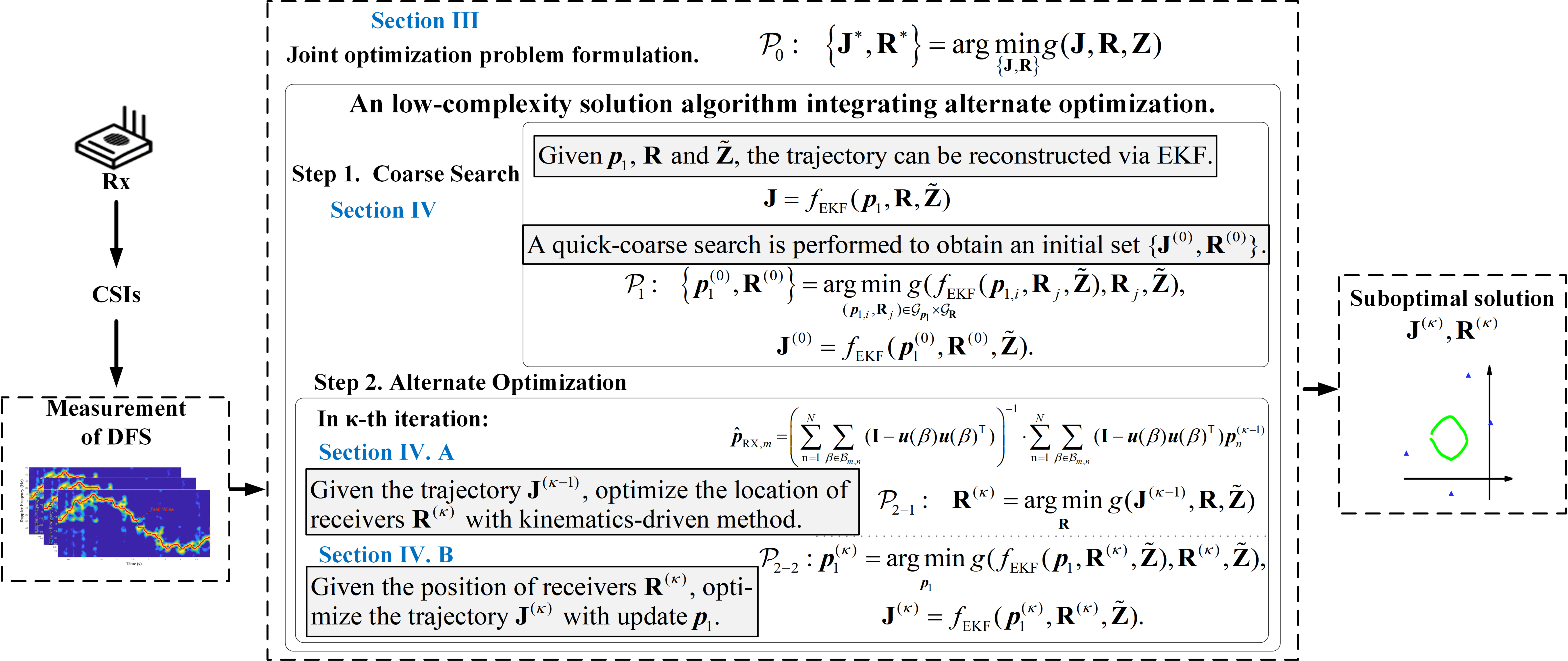}
    
    \caption{An workflow of low-complexity sub-optimal algorithm integrating alternate optimization.}
    \label{fig:workflow}
\end{figure*}

Note that given the locations of the receive stations $\mathbf{R}$ and the starting point of the trajectory $\boldsymbol{p}_1$, the reconstruction of the trajectory according to the measured Doppler frequencies can be conducted via the Extended Kalman Filter (EKF). It is essential to determine $\mathbf{R}$ and $\boldsymbol{p}_1$ in solving the Problem $\mathcal{P}_{0}$. However, the objective of the Problem $\mathcal{P}_{0}$ is non-convex with respect to $\mathbf{R}$ and $\boldsymbol{p}_1$. The gradient-based method may be trapped at local optimal solutions, which may be far from the true locations. Hence, a two-stage search method is proposed in WiSLAT, which consists of a {\it coarse search} from a number of candidates and a {\it fine adjustment} based on alternate optimization. An overall workflow of the 
proposed method is illustrated in \cref{fig:workflow}.

To avoid the potential trapping at local optimal solution deviating significantly from the true locations, it is necessary to start the optimization of the Problem $\mathcal{P}_0$ from a good initial solution of $\mathbf{J}$ and $\mathbf{R}$. Without any prior information on them, the coarse search stage will compare the MSE of multiple candidates. Particularly, denote $\mathcal{G}_{\boldsymbol{p}_1}$ and $\mathcal{G}_\mathbf{R}$, where $$\mathcal{G}_{\boldsymbol{p}_1}=\{\boldsymbol{p}_{1,i}\ |\ i=1,...,|\mathcal{G}_{\boldsymbol{p}_1}|\}$$ and $$\mathcal{G}_\mathbf{R}=\{\mathbf{R}_j\ |\ j=1,...,|\mathcal{G}_\mathbf{R}|\},$$ as the candidate sets of the starting point and receive stations' locations. 
The trajectory reconstructed from $\boldsymbol{p}_{1,i}$ and $\mathbf{R}_j$ according to the EKF can be expressed as

\begin{equation}
    \mathbf{J}_{i,j} = f_\mathrm{EKF}(\boldsymbol{p}_{1,i}, \mathbf{R}_j,\widetilde{\mathbf{Z}}),
    \label{con:EKF}
\end{equation}
where the function $f_\mathrm{EKF}$ represents the EKF-based trajectory reconstruction, as elaborated in \cref{app:EKFAppendix}. Hence, the initial solution of the staring point and receive stations' locations, denoted as  $\{\boldsymbol{p}_{1}^{(0)},\mathbf{R}^{(0)}\}$, can be obtained as
\begin{equation}
    \mathcal{P}_1:\quad\left\{\boldsymbol{p}_1^{(0)},\mathbf{R}^{(0)}\right\}=\mathop{\arg \min}\limits_{(\boldsymbol{p}_{1,i},\mathbf{R}_j) \in \mathcal{G}_{\boldsymbol{p}_1} \times \mathcal{G}_\mathbf{R}}g(\mathbf{J}_{i,j},\mathbf{R}_j, \widetilde{\mathbf{Z}}).
\end{equation}
Moreover, the initial solution of trajectory $\mathbf{J}^{(0)}$ is derived from $\boldsymbol{p}_1^{(0)}$ and $\mathbf{R}^{(0)}$ according to
\begin{equation}
    \mathbf{J}^{(0)} = f_\mathrm{EKF}(\boldsymbol{p}_{1}^{(0)}, \mathbf{R}^{(0)},\widetilde{\mathbf{Z}}).
\end{equation}

Given the initial solution $(\mathbf{J}^{(0)},\mathbf{R}^{(0)})$, their fine adjustment can be conducted by alternate optimization as follows. First, given the trajectory of the moving target, find the best locations of receive stations to minimize the MSE of Doppler measurements. Then, given the optimized locations of receive stations, continue to adjust the trajectory of the moving target.  Particularly, let $\mathbf{R}^{(\kappa)}$, $\boldsymbol{p}_1^{(\kappa)}$, and $\mathbf{J}^{(\kappa)}$ denote the receivers' locations, starting point, and trajectory after the $\kappa$-th iteration, respectively. The alternate optimization at the $\kappa$-th iteration, $\kappa = 1,2,\cdots$, can be decomposed as
\begin{equation}
\begin{aligned}
\mathcal{P}_{2-1} :   \mathbf{R}^{(\kappa)}&=\mathop{\arg \min}\limits_{\mathbf{R}}g(\mathbf{J}^{(\kappa-1)},\mathbf{R},\widetilde{\mathbf{Z}}),\\
\mathcal{P}_{2-2} :\boldsymbol{p}_1^{(\kappa)}&=\mathop{\arg \min}\limits_{\boldsymbol{p}_1}g(f_\mathrm{EKF}(\boldsymbol{p}_1, \mathbf{R}^{(\kappa)},\widetilde{\mathbf{Z}}),\mathbf{R}^{(\kappa)},\widetilde{\mathbf{Z}}).
\end{aligned}
\end{equation}
Moreover, given $\mathbf{R}^{(\kappa)}$ and $\boldsymbol{p}_1^{(\kappa)}$, the trajectory can be reconstructed as
\begin{equation}
    \mathbf{J}^{(\kappa)} = f_\mathrm{EKF}(\boldsymbol{p}_{1}^{(\kappa)}, \mathbf{R}^{(\kappa)},\widetilde{\mathbf{Z}}).
\end{equation}

In the following parts, the solution algorithms for the above two sub-problems are presented, respectively.

\subsection{Optimization of $\mathbf{R}$}

Given the trajectory of the moving target $\mathbf{J}^{(\kappa-1)}$ and the Doppler measurements $\widetilde{\mathbf{Z}}$, the search of the receive stations' locations in Problem $\mathcal{P}_{2-1}$ can actually be decoupled. Denote the search of the $m$-th receive station ($m=1,2,...,M$) as
\begin{equation}
     \mathcal{P}_{2-1}^m :   \boldsymbol{p}_{\mathrm{RX},m}^{(\kappa)}=\mathop{\arg \min}\limits_{\boldsymbol{p}_{\mathrm{RX},m}}g_m(\mathbf{J}^{(\kappa-1)},\boldsymbol{p}_{\mathrm{RX},m},\widetilde{\boldsymbol{z}}_m),
\end{equation}
where $\widetilde{\boldsymbol{z}}_m$ is the $m$-th row of measurement matrix $\widetilde{\mathbf{Z}}$,
\begin{equation}
    g_m(\mathbf{J}^{(\kappa-1)},\boldsymbol{p}_{\mathrm{RX},m},\widetilde{\boldsymbol{z}}_m)=\frac{1}{N-1}\|\widetilde{\mathbf{z}}_m -  \boldsymbol{z}_m(\mathbf{J}^{(\kappa-1)},\boldsymbol{p}_{\mathrm{RX},m})\|_\mathrm{F}^2
\end{equation}
is the MSE of the measurements at the Rx-m, and $\boldsymbol{z}_m(\mathbf{J}^{(\kappa-1)},\boldsymbol{p}_{\mathrm{RX},m})$ is the $m$-th row vector of matrix $\mathbf{Z}(\mathbf{J}^{(\kappa-1)},\mathbf{R})$. 

Without loss of generality, we focus on the optimization of $\boldsymbol{p}_{\mathrm{RX},m}$ in $\mathcal{P}_{2-1}^m$, which is still non-convex. Hence, the performance of the gradient-based search algorithms for $\boldsymbol{p}_{\mathrm{RX},m}$ depends heavily on the initial solution of search. In this part, a kinematics-driven method is first proposed to obtain its initial solution. Then, the Levenberg-Marquardt(LM) algorithm\cite{lm} is applied to update the initial solution iteratively.

First of all, we have the following conclusion on the direction of the $m$-th receive station with respect to the moving target at the $n$-th time instance ($\forall m,n$).

\begin{Lemma}[Potential Directions of Receive Station]\label{lemma:beta}
  As illustrated in \cref{fig:motionmodel}, given the position $\boldsymbol{p}_n$ and velocity $\boldsymbol{v}_n$ of the moving target at the $n$-th time instance, as well as the actual Doppler frequency $f^d_m(n)$ at the Rx-$m$, the angle rotating\footnote{The anticlockwise rotation leads to a positive angle, and clockwise rotation leads to a negative angle.} from the velocity  $\boldsymbol{v}_n$ to the vector $\boldsymbol{p}_{\mathrm{RX},m}-\boldsymbol{p}_n$, denoted as $\beta_{m,n}$, is within the following set of two elements\footnote{ Given a coarse AoA estimate obtained through state-of-the-art methods (e.g., SpotFi\cite{SpotFi}), the $\beta_{m,n}$ angles that fall outside the AoA range can be neglected.}:
  \begin{equation}
    \begin{aligned}
         &\mathcal{B}_{m,n}(f^d_{m}(n))=\\
         &\bigg\{ \pm \arccos{\bigg(\frac{1}{|\boldsymbol{v}_n|}(\lambda f^d_{m}(n)-\frac{(\boldsymbol{p}_\mathrm{TX} -\boldsymbol{p}_n)^\mathsf{T}}{|\boldsymbol{p}_\mathrm{TX} -  \boldsymbol{p}_n|}\boldsymbol{v}_n)\bigg)}\bigg\}.
    \end{aligned}
  \end{equation}
  Moreover, denote $\boldsymbol{u}(\boldsymbol{v}_n,\beta_{m,n})$ as the unit direction vector of $\boldsymbol{p}_{\mathrm{RX},m}-\boldsymbol{p}_n$. It can be expressed in terms of $\beta_{m,n}$ as 
  \begin{equation}
      \boldsymbol{u}(\boldsymbol{v}_n,\beta_{m,n}) = [\cos (\angle \boldsymbol{v}_n + \beta_{m,n}),\sin(\angle \boldsymbol{v}_n +\beta_{m,n})]^\mathsf{T}.
  \end{equation}
\end{Lemma}
\begin{proof} \it{
As shown in \cref{fig:motionmodel}, the Doppler frequency at the Rx-$m$ is due to the projection of $\boldsymbol{v}_n$ on two directions: $\boldsymbol{p}_{\mathrm{RX},m}-\boldsymbol{p}_n$ and $\boldsymbol{p}_{\mathrm{TX}}-\boldsymbol{p}_n$. Hence, 
$$\lambda f^d_{m}(n)=\frac{(\boldsymbol{p}_\mathrm{TX} -\boldsymbol{p}_n)^\mathsf{T}}{|\boldsymbol{p}_\mathrm{TX} -  \boldsymbol{p}_n|}\boldsymbol{v}_n + |\boldsymbol{v}_n| \cos \beta_{m,n}. $$ As a result, the conclusion can be obtained.}
\end{proof}

\begin{figure}[h]
    \centering
    \includegraphics[width=0.8\columnwidth]{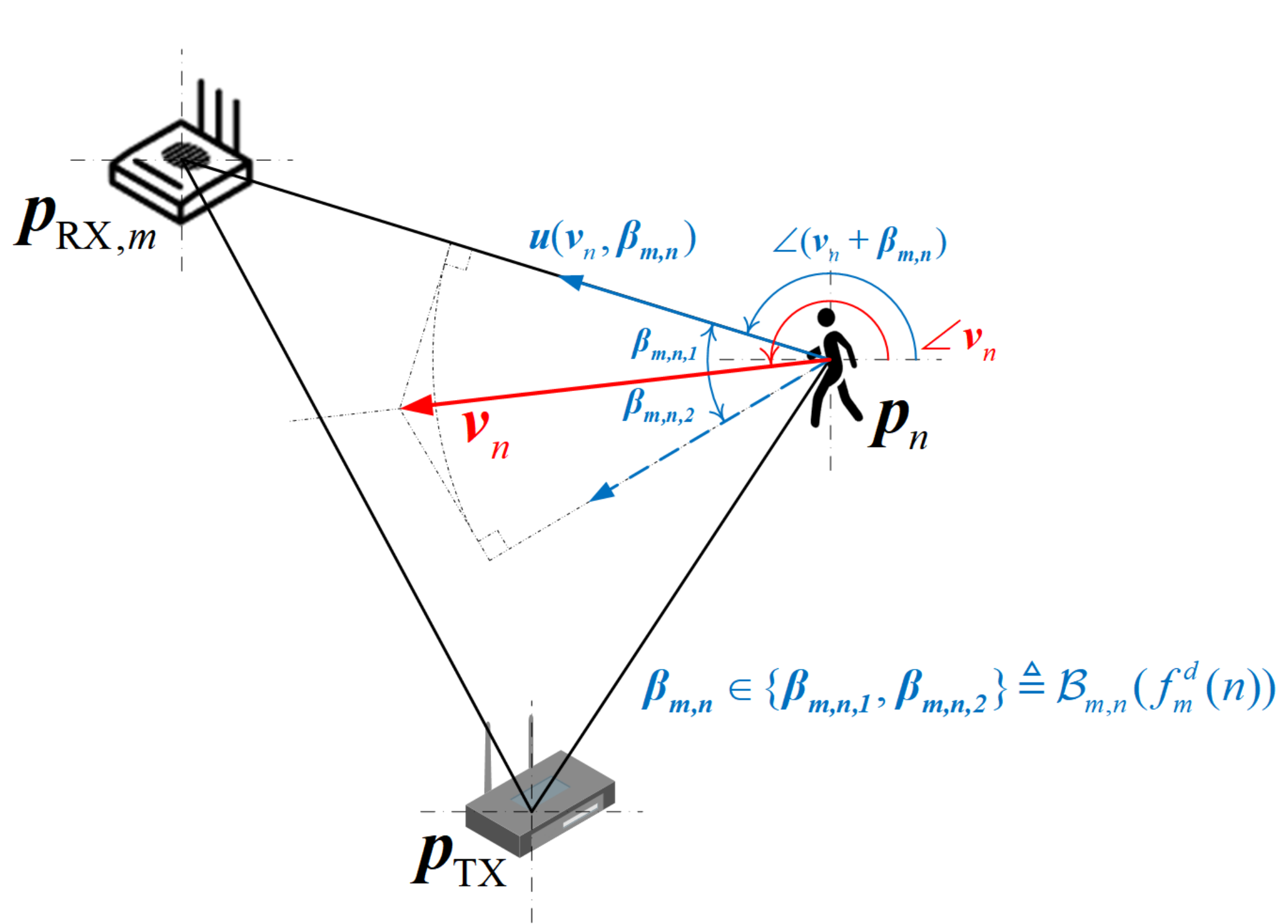}
    \caption{Kinematic diagram of a Tx-Rx pair and a moving target.}
    \label{fig:motionmodel}
\end{figure}

Hence, originated from the position $\boldsymbol{p}_n$, we can draw two rays with the directions of $\boldsymbol{u}(\boldsymbol{v}_n,\beta)$, $\beta \in \mathcal{B}_{m,n}$, such that the Rx-$m$ should be on one of the two rays. However, due to the noise in the Doppler detection and the error in the estimated trajectory, the actual location of Rx-$m$ may deviate from the above rays. Hence, we choose the following position, which has the minimum sum-distance to all the rays from $\{\boldsymbol{p}_1^{(\kappa-1)}, \boldsymbol{p}_2^{(\kappa-1)},..., \boldsymbol{p}_N^{(\kappa-1)}\}$, as the initial solution of Problem $\mathcal{P}_{2-1}^m$ \footnote{In order to suppress the computation complexity, we can choose a number of time instances, instead of all the $N$ time instances, to determine  $\hat{\boldsymbol{p}}_{\mathrm{RX},m}$.}:
  \begin{equation}
  \begin{aligned}
    \hat{\boldsymbol{p}}_{\mathrm{RX},m}\! =  \! \left(\sum^N_{\substack{n=1 \\\beta\in\mathcal{B}_{m,n}(\tilde{f^d_m}(n))}}(\mathbf{I}-\boldsymbol{u}(\boldsymbol{v}_n^{(\kappa-1)},\beta)\boldsymbol{u}(\boldsymbol{v}_n^{(\kappa-1)},\beta)^\mathsf{T})\right)^{-1}\\
    \cdot\sum^N_{\substack{n=1 \\\beta\in\mathcal{B}_{m,n}(\tilde{f^d_m}(n))}}(\mathbf{I}-\boldsymbol{u}(\boldsymbol{v}_n^{(\kappa-1)},\beta)\boldsymbol{u}(\boldsymbol{v}_n^{(\kappa-1)},\beta)^\mathsf{T})\boldsymbol{p}_n^{(\kappa-1)},
  \end{aligned}
  \end{equation}

The location of the $m$-th receive station can then be updated according to the LM algorithm. First, we have the following conclusion on the gradient of the MSE with respect to the location of the Rx-$m$.

\begin{Lemma}[Gradient of $g_m$]\label{lemma:gm}
  Given the trajectory of the moving target $\mathbf{J}^{(\kappa-1)}$ and the Doppler measurements $\widetilde{\mathbf{z}}_m$, the first and second orders of gradients of the $g_m$ (with respect to $\boldsymbol{p}_{\mathrm{RX},m}$) can be expressed as 
  \begin{equation}
      \nabla g_m(\mathbf{J}^{(\kappa-1)},\boldsymbol{p}_{\mathrm{RX},m},\widetilde{\boldsymbol{z}}_m) = 2\mathbf{G}_m^{\mathsf{T}}\boldsymbol{e}_{m}
      \label{con:gm_1stgrad}
  \end{equation}
  and 
  \begin{equation}
      \nabla^2 g_m(\mathbf{J}^{(\kappa-1)},\boldsymbol{p}_{\mathrm{RX},m},\widetilde{\boldsymbol{z}}_m)\approx 2\mathbf{G}_m^{\mathsf{T}}\mathbf{G}_m
  \end{equation}
  respectively, where the error vector $$\boldsymbol{e}_{m}=[\widetilde{\boldsymbol{z}}_m -  \boldsymbol{z}_m(\mathbf{J}^{(\kappa-1)},\boldsymbol{p}_{\mathrm{RX},m})]^{\mathsf{T}}$$ and $\mathbf{G}_m$ denotes the Jacobian matrix of $\boldsymbol{e}_{m}$ with respect to $\boldsymbol{p}_{\mathrm{RX},m}$. 
\end{Lemma}
\begin{proof}\it{
 The proof and expression of $\mathbf{G}_m$ is provided in \cref{app:gradient1}.}
\end{proof}

As a result, following the LM algorithm, the $(i\!+\!1)$-th update on $\boldsymbol{p}_{\mathrm{RX},m}^{(\kappa)}$, denoted as $\boldsymbol{p}_{\mathrm{RX},m}^{(\kappa),i+1}$ ($i=0,1,2,...,$), can be written as
\begin{equation}
    \boldsymbol{p}_{\mathrm{RX},m}^{(\kappa),i+1} = \boldsymbol{p}_{\mathrm{RX},m}^{(\kappa),i}  -({\mathbf{G}_m^{(\kappa),i}}^{\mathsf{T}}\mathbf{G}_m^{(\kappa),i}+\mu\mathbf{I})^{-1}{\mathbf{G}_m^{(\kappa),i}}^{\mathsf{T}}\boldsymbol{e}_{m}^{(\kappa),i},
\end{equation}
where $\mu$ is a damping factor, $\mathbf{G}_m^{(\kappa),i}$ and $\boldsymbol{e}_{m}^{(\kappa),i}$ are the Jacobian matrix and error vector defined in Lemma \ref{lemma:gm} with $\boldsymbol{p}_{\mathrm{RX},m}=\boldsymbol{p}_{\mathrm{RX},m}^{(\kappa),i} $, and the initial solution $\boldsymbol{p}_{\mathrm{RX},m}^{(\kappa),0}=\hat{\boldsymbol{p}}_{\mathrm{RX},m}$. Denote $\boldsymbol{p}_{\mathrm{RX},m}^{(\kappa)}$ as the solution after convergence. The solution for the Problem $\mathcal{P}_{2-1}$ can be written as
$$
    \mathbf{R}^{(\kappa)} = [\boldsymbol{p}_{\mathrm{RX},1}^{(\kappa)},\boldsymbol{p}_{\mathrm{RX},2}^{(\kappa)},...,\boldsymbol{p}_{\mathrm{RX},M}^{(\kappa)}].
$$

\subsection{Optimization of $\boldsymbol{p}_1$}

As shown in \cref{con:Jpv}, the trajectory can be equivalently represented by the starting point $\boldsymbol{p}_1$ and the following velocities $(\boldsymbol{v}_1,...,\boldsymbol{v}_{N-1})$. Note that the adjustment of $\boldsymbol{p}_1$ will also lead to the update of the following velocities according to the EKF algorithm. Since there is no closed-form expression of the following velocities $(\boldsymbol{v}_1,...,\boldsymbol{v}_{N-1})$ versus the starting point $\boldsymbol{p}_1$ in EKF filtering, it is proposed to decompose the optimization of $\boldsymbol{p}_1$ and $(\boldsymbol{v}_1,...,\boldsymbol{v}_{N-1})$ into two iterative steps. For the elaboration convenience, the iteration between the Problem $\mathcal{P}_{2-1}$ and $\mathcal{P}_{2-2}$ is referred to as the outer iteration, and the iteration solving the Problem $\mathcal{P}_{2-2}$ is referred to as inner iteration in this part. 

Denote the trajectory obtained from the last ($(\kappa\!-\!1)$-th) outer iteration as $\mathbf{J}^{(\kappa-1)}\triangleq(\boldsymbol{p}_1^{(\kappa-1)},\boldsymbol{v}_1^{(\kappa-1)},...,\boldsymbol{v}_{N-1}^{(\kappa-1)})$. The iterative algorithm for the Problem $\mathcal{P}_{2-2}$ is elaborated below:
\begin{itemize}
    \item[1.] Initialize the inner iteration index as $j=1$, and the initial solution of inner iteration as $\mathbf{J}^{(\kappa,0)}=\mathbf{J}^{(\kappa-1)}$.
    \item[2.] In the $j$-th inner iteration, given the velocities of last inner iteration $(\boldsymbol{v}_1^{(\kappa,j-1)},...,\boldsymbol{v}_{N-1}^{(\kappa,j-1)})$, update the starting point $\boldsymbol{p}_1^{(\kappa,j)}$ according to the LM algorithm as
    \begin{equation}
        \boldsymbol{p}_1^{(\kappa,j)}=\boldsymbol{p}_1^{(\kappa,j-1)}-\Delta \boldsymbol{p}_1^{(\kappa,j-1)},
    \end{equation}
    where 
    $$
    \begin{aligned}
    \Delta \boldsymbol{p}_1^{(\kappa,j-1)}=& ({\mathbf{G}_{\boldsymbol{p}_1}^{(\kappa,j-1)}}^\mathsf{T}\mathbf{G}_{\boldsymbol{p}_1}^{(\kappa,j-1)}+\mu \mathbf{I})^{-1}{\mathbf{G}_{\boldsymbol{p}_1}^{(\kappa,j-1)}}^\mathsf{T}\nonumber\\
    &\cdot\text{vec}(\mathbf{E}_{\mathbf{Z}}^{(\kappa,j-1)}),\nonumber
    \end{aligned}
    $$
    $\mathbf{E}_{\mathbf{Z}}^{(\kappa,j-1)}=\widetilde{\mathbf{Z}} -  \mathbf{Z}(\mathbf{J}^{(\kappa,j-1)},\mathbf{R}^{(\kappa)})$ is the matrix of measurement errors, $\text{vec}(\mathbf{E}_{\mathbf{Z}}^{(\kappa,j-1)})$ represents the vectorization of $\mathbf{E}_{\mathbf{Z}}^{(\kappa,j-1)}$, and
    $\mathbf{G}_{\boldsymbol{p}_1}^{(\kappa,j-1)}$ is the Jacobian matrix of $\mathbf{E}_{\mathbf{Z}}^{(\kappa,j-1)}$ with respect to $\boldsymbol{p}_1^{(\kappa,j-1)}$ at the $j$-th inner iteration\footnote{The derivation of Jacobian matrix $\mathbf{G}_{\boldsymbol{p}_1}^{(\kappa,j-1)}$ is similar to that in \cref{lemma:beta}, which is omitted here due to page limitation.}.   
    \item[3.] Given the starting point $\boldsymbol{p}_1^{(\kappa,j)}$, reconstruct the velocities $(\boldsymbol{v}_1^{(\kappa,j)},...,\boldsymbol{v}_{N-1}^{(\kappa,j)})$ via EKF, thus
    \begin{eqnarray}
       \mathbf{J}^{(\kappa,j)}&\triangleq& (\boldsymbol{p}_1^{(\kappa,j)},\boldsymbol{v}_1^{(\kappa,j)},...,\boldsymbol{v}_{N-1}^{(\kappa,j)})\nonumber\\ &=& f_{EKF}(\boldsymbol{p}_1^{(\kappa,j)},\mathbf{R}^{(\kappa)},\widetilde{\mathbf{Z}}).
    \end{eqnarray}
    \item[4.] Let $j=j+1$, jump to the second step, until convergence or a maximum number of iterations reaches. 
\end{itemize}

\begin{Remark}\it{
In the above algorithm, the velocities of the moving target since the first time interval $(\boldsymbol{v}_1^{(\kappa,j-1)},...,\boldsymbol{v}_{N-1}^{(\kappa,j-1)})$ determine the shape of the trajectory, and the starting point determine where to pin the shape. Hence intuitively, the inner iteration first fix the shape of trajectory and determine where to put the trajectory, then reshape the trajectory according to the starting point.

Note that given the velocities, the partial derivative of the MSE with respective to the starting point can be derived with closed-form expression. Hence, the decomposition in the above algorithm could significantly suppress the computation complexity of gradient-based search algorithm.}    
\end{Remark}

\section{EXPERIMENTS AND DISCUSSIONS}
\subsection{Experimental Settings}
\begin{figure}[h]
    \centering
    \includegraphics[width=0.75\columnwidth]{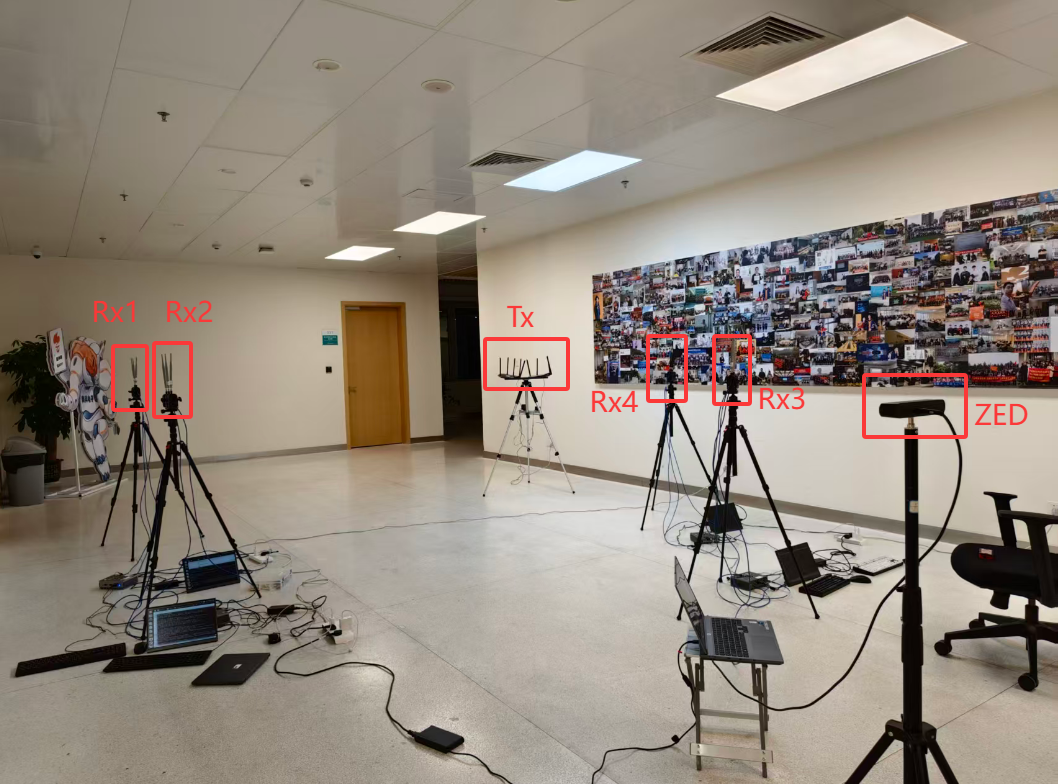}
    \caption{Experiment setup in an indoor corridor scenario.}
    \label{fig:setup}
\end{figure}
\begin{figure*}[!h]
    \centering
    \subfloat[Circle]{\includegraphics[width=1.5in]{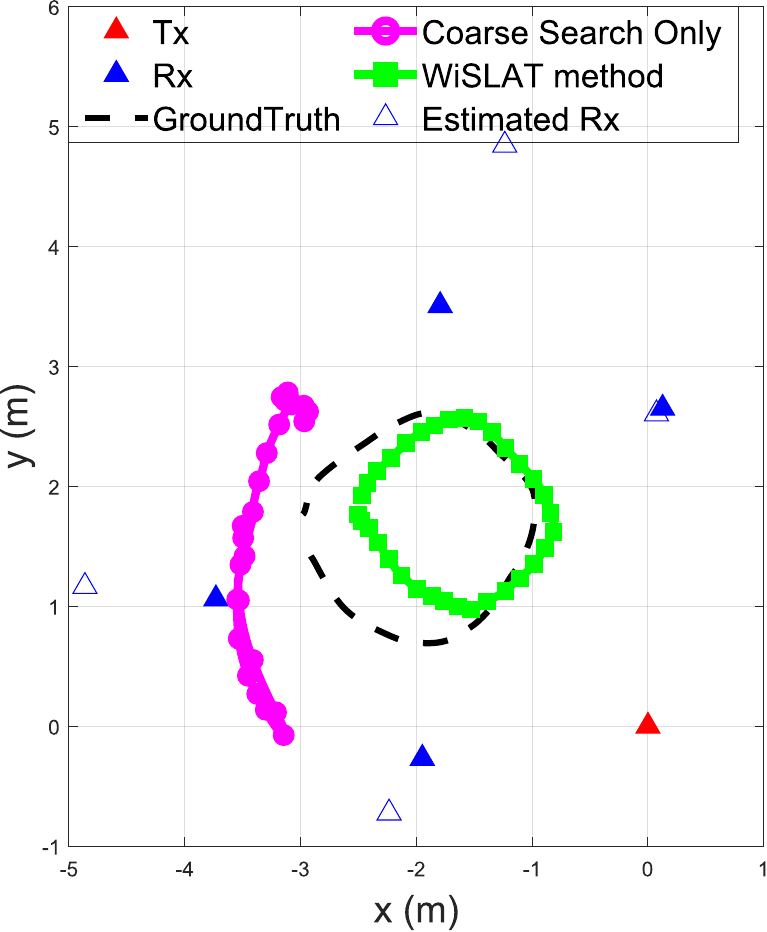}%
    \label{traj1}}
    \hfil
    \subfloat[Square]{\includegraphics[width=1.5in]{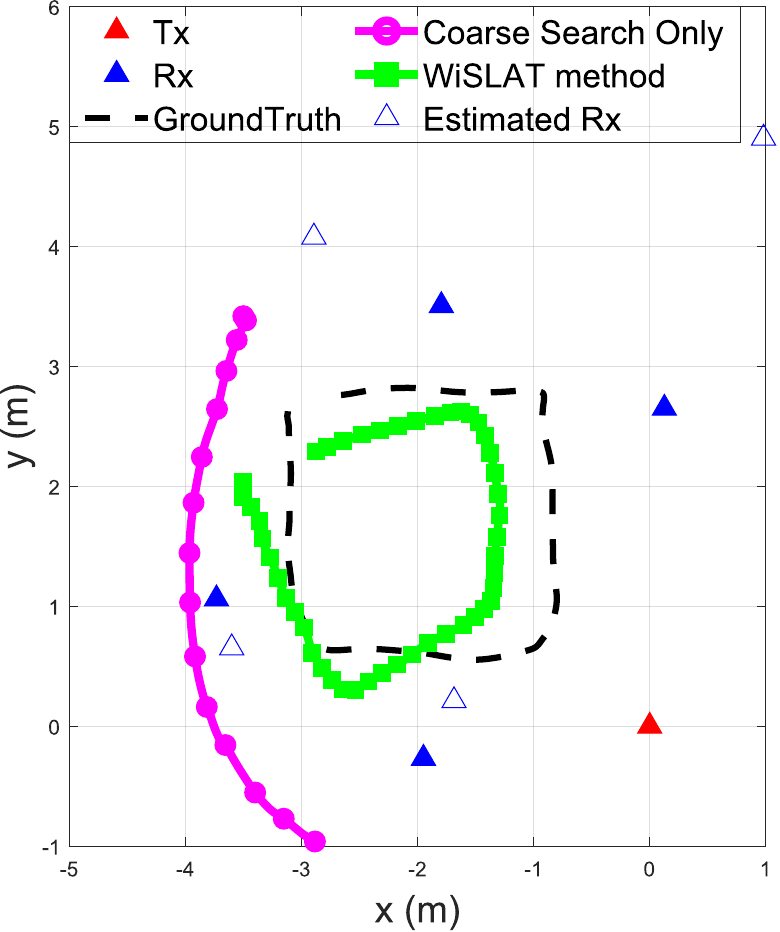}%
    \label{traj2}}
    \hfil
    \subfloat[Triangle]{\includegraphics[width=1.5in]{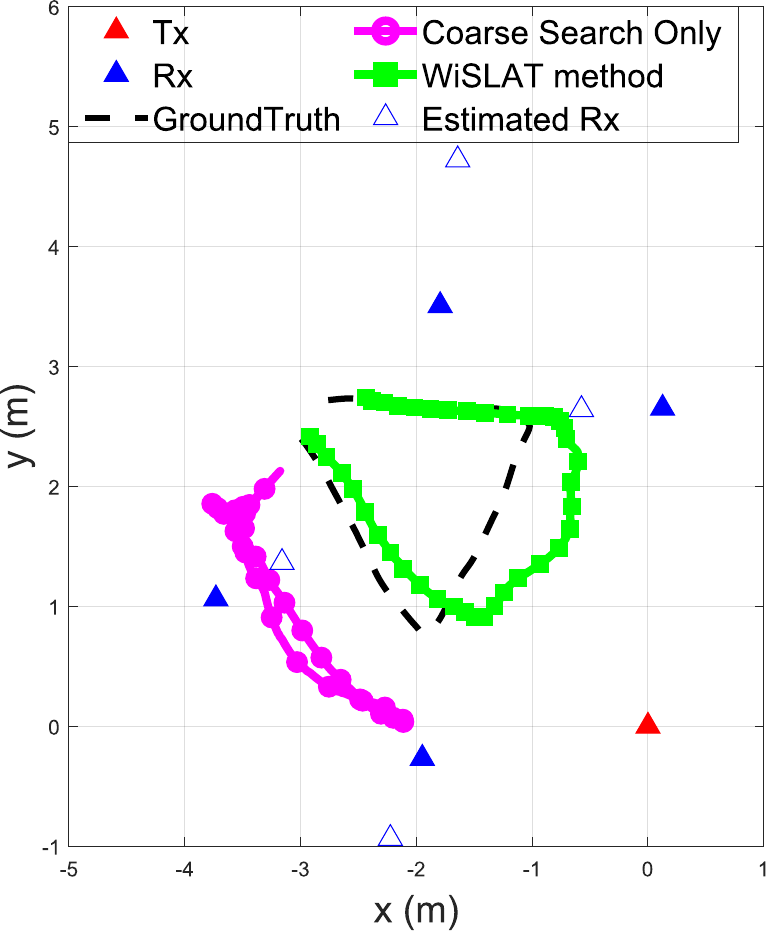}%
    \label{traj3}}
    \hfil
    \subfloat[CDF of tracking error]{\includegraphics[width=1.2in]{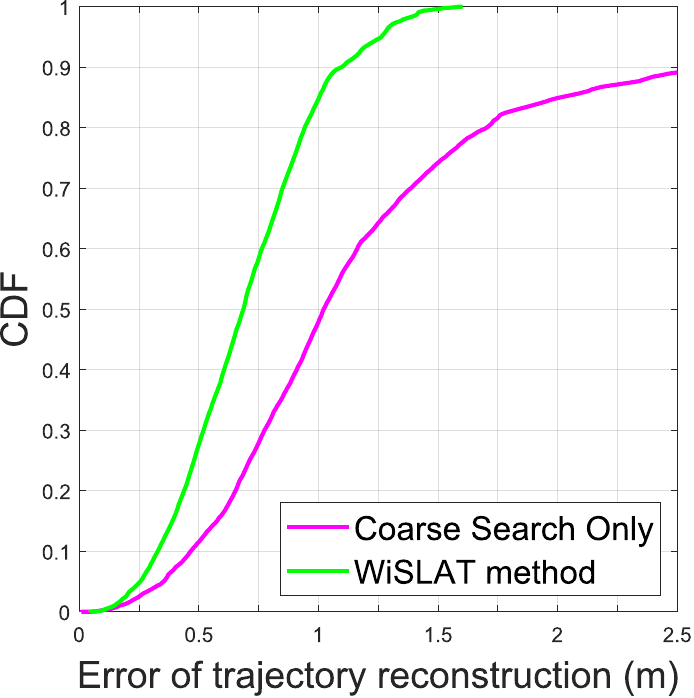}%
    \label{cdf_1}}
    \hfil
    \subfloat[CDF of Localization error]{\includegraphics[width=1.2in]{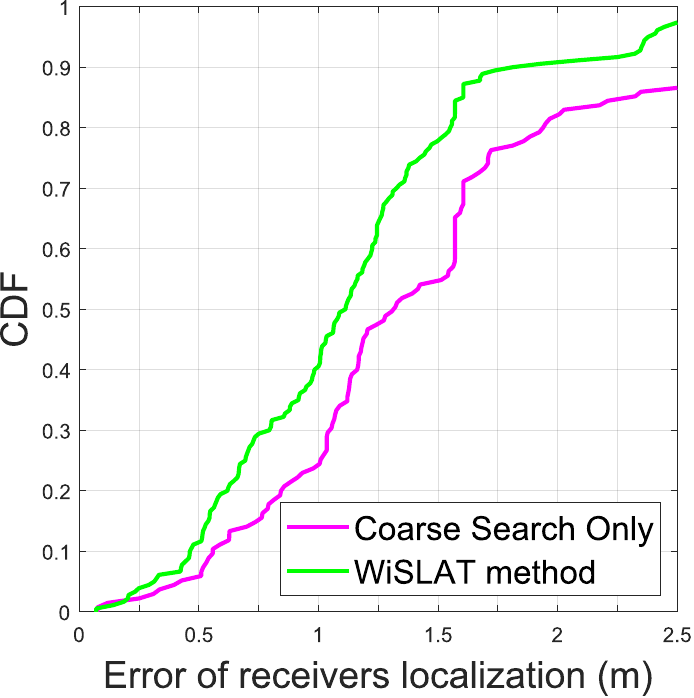}%
    \label{cdf_2}}
    \hfil
    \caption{(a-c) Illustration of WiSLAT trajectory reconstruction and receivers localization in three samples.\\(d-e) CDF of (a) reconstructed trajectory error and (b) receivers localization error.}
    \label{fig:res}
\end{figure*}

We used a TP-LINK TL-XDR1520 router as the AP and four mini PCs equipped with Intel 5300 NIC as the Rx, each equipped with three antennas. The four Rx were randomly distributed around the AP, and their locations were unknown in the trajectory tracking. The Linux 802.11n CSI Tool \cite{csitool} was used to collect CSI data from Rx. The carrier frequency of the AP signal is 5.24 GHz with a 20 MHz bandwidth. Since the Doppler 
frequencies introduced by human movement in indoor environments does not exceed 50Hz at 5.24 GHz, a sampling rate of 100Hz (i.e., a 10 ms interval) was used for CSI collection. 

The moving target is one volunteer. We employed depth camera to record the ground truth of the volunteer's trajectories, where the midpoint of the chest and abdomen was chosen as representative of the volunteer's position. The height is approximately 1.3 m. 

Since the Doppler detection method relies on the dominant LoS path between the AP and the receive stations, it is necessary to exclude the receive station with blocked LoS path in Doppler detection and EKF-based trajectory reconstruction. Particularly, a significant attenuation in received signal strength is treated as a LoS blockage. Since there are four receive stations in the experiment, there is always at least three available receive stations in sensing.

The experiment scenario is shown in \cref{fig:setup}. We evaluated the simultaneous localization and track (SLAT) performance of our proposed method in an indoor corridor. The distances between the AP and the receive stations are $2-4\,\text{m}$, and the entire sensing area is approximately $5\text{m} \times 5\text{m}$. 3 different shapes of trajectories, including circular, square, and triangular, were sensed. Each trajectory shape was tested by 3 different volunteers, and each volunteer repeat each shape by 5 times. Hence, totally 45 trajectories were sensed in the experiment.

\subsection{Experimental Results}

\cref{traj1,traj2,traj3} illustrate the reconstructed trajectories and the estimated locations of the receive stations versus the ground truth in three examples, including the circular, square, and triangular trajectories. The purple and green trajectories refer to the results of coarse search and the overall proposed WiSLAT method. It can be observed that fine optimization after coarse search is necessary, as the purple curves deviate from the ground truth a lot and the green ones are close to the ground truth. Moreover, in most of the cases, the location errors of the receive stations are less than 1 m.

In order to see the overall performance of the proposed WiSLAT method, the statistics of all 45 trajectories are illustrated in \cref{cdf_1,cdf_2}, where Cumulative Distribution Functions (CDFs) of tracking error and localization error counting the 45 trajectories are plotted. It reveals that, with coarse search only, the 50\% error for tracking is 1.02 m whereas the 50\% error for Rx localization is 1.3 m. On the other hand, with the complete WiSLAT method, the 50\% error for tracking is 0.68 m whereas the 50\% error for Rx localization is 1.07 m.
% \begin{figure}[h]
%     \centering
%     \subfloat[]{\includegraphics[width=0.4\columnwidth]{cdf1.pdf}%
%     \label{cdf_1}}
%     \hfil
%     \subfloat[]{\includegraphics[width=0.4\columnwidth]{cdf2.pdf}%
%     \label{cdf_2}}
%     \caption{CDF of (a) reconstructed trajectory error and (b) receivers localization error.}
%     \label{fig:cdf}
% \end{figure}

The above results show that our system can correctly reconstruct the trajectory and detect the receive stations' locations. In fact, two receive stations are sufficiently to reconstruct two-dimensional trajectories, if their locations are known in advance. On the other hand, without the knowledge on the receive stations' locations, three receive stations become necessary in simultaneous localization and tracking. The additional dimension of Doppler measurements is sufficient to compensate the uncertainty on the receive stations' locations. Considering the requirement of LoS link, we deployed four receivers to ensure at least three LoS links exist at any time. In fact, more Doppler measurements could provide more redundancy against the measurement noise.

Finally, note that the localization accuracy is generally worse than the trajectory tracking accuracy in some trajectories. It can be observed that different trajectories lead to different localization accuracy for different receive stations. This could motivate a joint design method, which average the locations of the receive stations in all 45 trajectories. It can be observed from \cref{fig:ave_Rx} that, by simply averaging the estimates across all trajectories, diverging errors in Rx localization are effectively mitigate.
\begin{figure}[h]
    \centering
    \includegraphics[width=0.75\columnwidth]{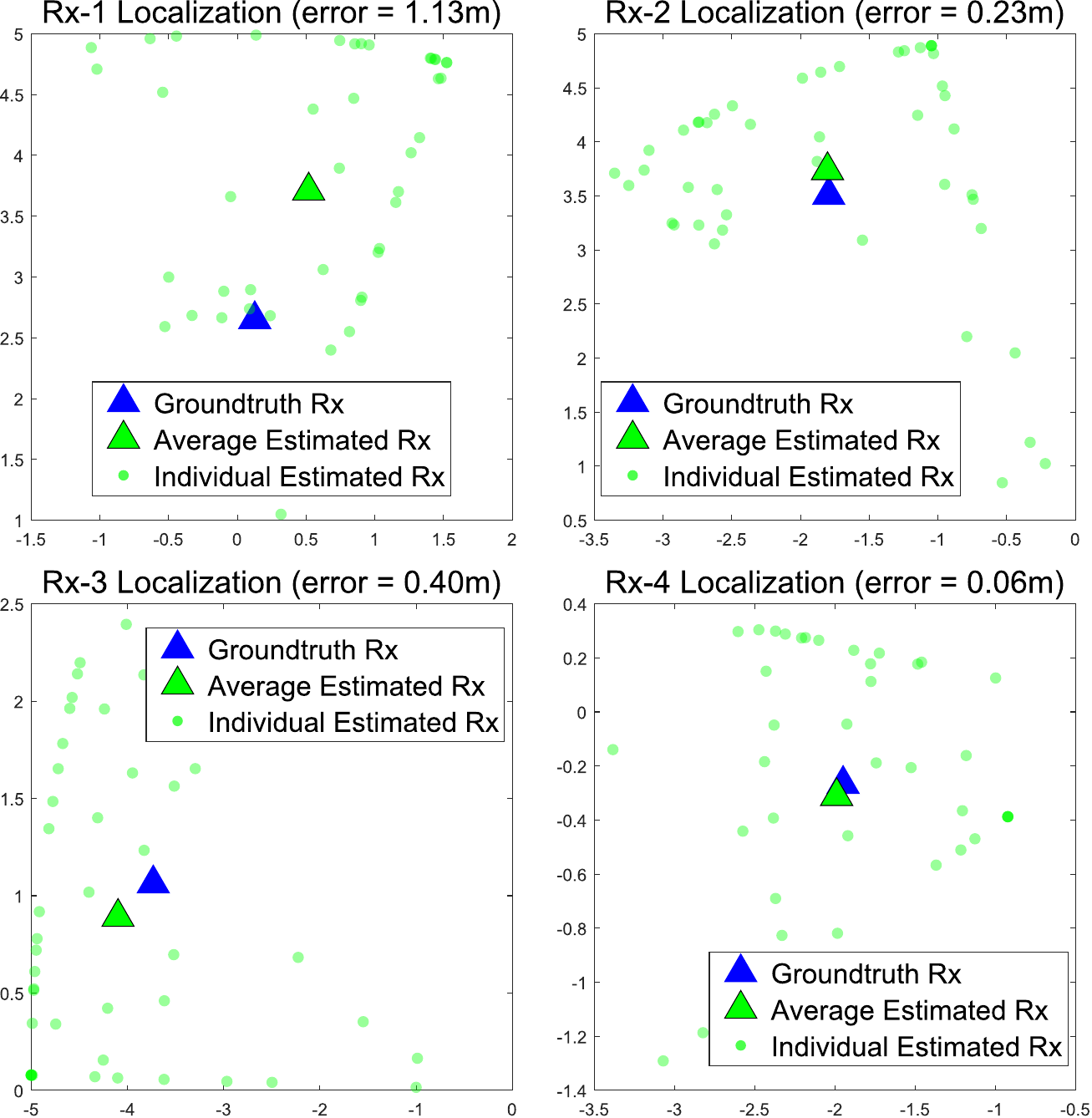}
    \caption{Final estimated Rx after averaging multiple observation.}
    \label{fig:ave_Rx}
\end{figure}

\section{Conclusion}
In this paper, a simultaneous localization and tracking method, namely WiSLAT, is proposed for a Wi-Fi communication system with at least three receive stations. It relies on the Doppler measurements of one single moving target at the receive stations to reconstruct the trajectory of the moving target, and meanwhile estimate the locations of the receive stations. The overall tracking and localization problem is formulated as an MMSE problem, and a two-stage optimization method is proposed to solve it with a low computational complexity. The performance of the WiSLAT is evaluated by experiment in a indoor scenario. It is shown that the median of the tracking and localization errors are 0.68m and 1.07m. Finally, by averaging the estimated locations of the receive stations, the localization accuracy is suppressed to 0.45m on average. This motivates a joint simultaneous localization and tracking design with multiple trajectories, which would be left for our future study.
 
\appendices
\section{EKF for Trajectory Reconstruction}\label{app:EKFAppendix}

Let $\boldsymbol{p}_n=[x_n, y_n]^\mathsf{T}$ and $\boldsymbol{v}_n=[v_{x_n}, v_{y_n}]^\mathsf{T}$ be the position and velocity of the moving target at the $n$-th time instance. Define the state vector at the $n$-th time instance as $\boldsymbol{s}_n = [x_n, y_n, v_{x_n}, v_{y_n}]^\mathsf{T}$, encompassing position and velocity. The state transition can be written as
\begin{equation}
    \boldsymbol{s}_{n} = \mathbf{F} \boldsymbol{s}_{n-1}+\boldsymbol{w}_{n} , 
\end{equation}
where $\mathbf{F}$ denotes the state transition matrix, and $\boldsymbol{w}_n\sim \mathcal{N}(0,\mathbf{W})$ is the process noise. The state transition matrix $\mathbf{F}$ is defined as
\begin{equation}
\mathbf{F} = \begin{bmatrix}
    1 & 0 & \Delta t & 0 \\
    0 & 1 & 0 &  \Delta t\\
    0 & 0 & 1 & 0 \\
    0 & 0 & 0 & 1 \\
    \end{bmatrix},
\end{equation}
and the covariance matrix $\mathbf{W}$ is defined as
\begin{equation}
    \mathbf{W} = \begin{bmatrix}
    \sigma_x^2 \frac{\Delta t^3}{3} & 0 & \sigma_x^2 \frac{\Delta t^2}{2} & 0 \\
    0 & \sigma_y^2 \frac{\Delta t^3}{3} & 0 & \sigma_y^2 \frac{\Delta t^2}{2} \\
    \sigma_x^2 \frac{\Delta t^2}{2} & 0 & \sigma_{v_x}^2 \Delta t & 0 \\
    0 & \sigma_y^2 \frac{\Delta t^2}{2} & 0 & \sigma_{v_y}^2 \Delta t \\
    \end{bmatrix}.
\end{equation}

Moreover, the measurement vector $\tilde{\boldsymbol{z}}_n$ at the $n$-th time instance is defined as
\begin{equation}
    \tilde{\boldsymbol{z}}_n = \boldsymbol{h}(\boldsymbol{s}_{n})+\boldsymbol{u}_{n},
\end{equation}
where 
\begin{equation}
    \boldsymbol{h}(\boldsymbol{s}_{n}) = 
    \begin{bmatrix}
    \frac{1}{\lambda}\left(\frac{\boldsymbol{p}_\mathrm{TX} - \boldsymbol{p}_n}{|\boldsymbol{p}_\mathrm{TX} - \boldsymbol{p}_n|} + \frac{\boldsymbol{p}_{\mathrm{RX},1} -\boldsymbol{p}_n}{|\boldsymbol{p}_{\mathrm{RX},1} - \boldsymbol{p}_n|}\right) \boldsymbol{v}_n\\
    \vdots\\
    \frac{1}{\lambda}\left(\frac{\boldsymbol{p}_\mathrm{TX} - \boldsymbol{p}_n}{|\boldsymbol{p}_\mathrm{TX} - \boldsymbol{p}_n|} + \frac{\boldsymbol{p}_{\mathrm{RX},M} -\boldsymbol{p}_n}{|\boldsymbol{p}_{\mathrm{RX},M} - \boldsymbol{p}_n|}\right) \boldsymbol{v}_n
    \end{bmatrix}\ ,
\end{equation}
is the nonlinear measurement function, and  $\boldsymbol{u}_n\sim\mathcal{N}(0,\mathbf{U})$ is the measurement noise with $\mathbf{U} = \sigma_{f^d} \cdot\mathbf{I}$. 

The trajectory reconstruction procedure employing EKF consists of two iterative steps: state prediction and state update. 
Let $\boldsymbol{s}_{n|n-1}$ and $\mathbf{P}_{n|n-1}$ denote predicted system state and the corresponding covariance matrix of the $n$-th time instance at the $(n-1)$-th time instance, and $\boldsymbol{s}_{n|n}$ and $\mathbf{P}_{n|n}$ denote the updated ones at the $n$-th time instance based on latest observations. We have the state prediction as
\begin{equation}
    \begin{cases}
    \boldsymbol{s}_{n|n-1}=\mathbf{F}\boldsymbol{s}_{n-1|n-1}\\
    \mathbf{P}_{n|n-1}=\mathbf{F}\mathbf{P}_{n-1|n-1}\mathbf{F}^{\mathsf{T}}+\mathbf{W}
\end{cases},
\end{equation}
and state update as
\begin{equation}
    \begin{cases}
    \mathbf{s}_{n|n}=\mathbf{s}_{n|n-1} + \mathbf{K}_n[\tilde{\boldsymbol{z}}_n - \boldsymbol{h}(\boldsymbol{s}_{n|n-1})]\\
    \mathbf{P}_{n|n}=(\mathbf{I}-\mathbf{K}_n\mathbf{D}_n)\mathbf{P}_{n|n-1}
\end{cases} ,
\end{equation}
where the Kalman gain matrix $\mathbf{K}_n$ is given by
\begin{equation}
    \mathbf{K}_n = \mathbf{P}_{n|n-1}\mathbf{D}_n^\mathsf{T}(\mathbf{D}_n\mathbf{P}_{n|n-1}\mathbf{D}_n^\mathsf{T} + \mathbf{U})^{-1}.
\end{equation}
Moreover, the Jacobian matrix $\mathbf{D}_n$ of $\boldsymbol{h}(s_{n|n-1})$ is written as
\begin{equation}
    \mathbf{D}_n = \begin{bmatrix}
    \frac{\partial f^d_{1}(n)}{\partial x_n}&\frac{\partial f^d_{1}(n)}{\partial y_n}&\frac{\partial f^d_{1}(n)}{\partial v_{x_n}}&\frac{\partial f^d_{1}(n)}{\partial v_{y_n}}\\
    \vdots&\vdots&\vdots&\vdots\\
    \frac{\partial f^d_{M}(n)}{\partial x_n}&\frac{\partial f^d_{M}(n)}{\partial y_n}&\frac{\partial f^d_{M}(n)}{\partial v_{x_n}}&\frac{\partial f^d_{M}(n)}{\partial v_{y_n}}\\
    \end{bmatrix},
\end{equation}
where the partial derivatives are given by
\begin{equation}
    \begin{cases}
    \frac{\partial f^d_{m}(n)}{\partial x_n}=\!\!\!\! & \frac{1}{\lambda}\{-v_{x_n}[\frac{(y_n-y_\mathrm{TX})^2}{|\boldsymbol{p}_n-\boldsymbol{p}_\mathrm{TX}|^3}+\frac{(y_n-y_{\mathrm{RX},m})^2}{|\boldsymbol{p}_n-\boldsymbol{p}_{\mathrm{RX},m}|^3}] \\
     &\!\!\!\!+v_{y_n}[\frac{(y_n-y_\mathrm{TX})(x_n-x_\mathrm{TX})}{|\boldsymbol{p}_n-\boldsymbol{p}_\mathrm{TX}|^3}+\frac{(y_n-y_{\mathrm{RX},m})(x_n-x_{\mathrm{RX},m})}{|\boldsymbol{p}_n-\boldsymbol{p}_{\mathrm{RX},m}|^3}]\}, \\
    \frac{\partial f^d_{m}(n)}{\partial y_n}=\!\!\!\! &
    \frac{1}{\lambda}\{v_{x_n}[\frac{(x_n-x_\mathrm{TX})(y_n-y_\mathrm{TX})}{|\boldsymbol{p}_n-\boldsymbol{p}_\mathrm{TX}|^3}+\frac{(x_n-x_{\mathrm{RX},m})(y_n-y_{\mathrm{RX},m})}{|\boldsymbol{p}_n-\boldsymbol{p}_{\mathrm{RX},m}|^3}] \\
     &\!\!\!\!-v_{y_n}[\frac{(x_n-x_\mathrm{TX})^2}{|\boldsymbol{p}_n-\boldsymbol{p}_\mathrm{TX}|^3}+\frac{(x_n-x_{\mathrm{RX},m})^2}{|\boldsymbol{p}_n-\boldsymbol{p}_{\mathrm{RX},m}|^3}]\}, \\
     \frac{\partial f^d_{m}(n)}{\partial v_{x_n}}=\!\!\!\! & \frac{1}{\lambda}[\frac{x_\mathrm{TX}-x_n}{|\boldsymbol{p}_n-\boldsymbol{p}_\mathrm{TX}|}+\frac{x_{\mathrm{RX},m}-x_n}{|\boldsymbol{p}_n-\boldsymbol{p}_{\mathrm{RX},m}|}],\\
     \frac{\partial f^d_{m}(n)}{\partial v_{y_n}}=\!\!\!\! & \frac{1}{\lambda}[\frac{y_\mathrm{TX}-y_n}{|\boldsymbol{p}_n-\boldsymbol{p}_\mathrm{TX}|}+\frac{y_{\mathrm{RX},m}-y_n}{|\boldsymbol{p}_n-\boldsymbol{p}_{\mathrm{RX},m}|}].\\
    \end{cases}
\label{con:jacobian}
\end{equation}

\section{Gradient Descent for $\boldsymbol{p}_{\mathrm{RX},m}$}\label{app:gradient1}

The gradient in \cref{con:gm_1stgrad} is straightforward. Moreover, the Jacobian matrix $\mathbf{G}_m$ of $\boldsymbol{e}_{m}$ is defined as 
\begin{equation}
    \mathbf{G}_m = \frac{\partial \boldsymbol{e}_{m}}{\partial\boldsymbol{p}_{\mathrm{RX},m}} = \begin{bmatrix}
    \vdots&\vdots\\
    \frac{\partial  \boldsymbol{e}_{m}(n)}{\partial x_{\mathrm{RX},m}}&\frac{\partial  \boldsymbol{e}_{m}(n)}{\partial y_{\mathrm{RX},m}}\\
    \vdots&\vdots
    \end{bmatrix}\in \mathbb{R}^{(N-1)\times 2}.
\end{equation}
Denote $ \mathbf{G}_m({n,j})$ as the $(n,j)$-th entry ($n=1,2,...,N-1;j=1,2$) of the Jacobian matrix $\mathbf{G}_m$, they can be derived as
\begin{equation}
    \begin{aligned}
    \mathbf{G}_m({n,1}) &= -\frac{\partial f^d_m(n)}{\partial x_{\mathrm{RX},m}}=\frac{|(\boldsymbol{p}_n-\boldsymbol{p}_{\mathrm{RX},m})\times\boldsymbol{v}_n|}{\lambda|\boldsymbol{p}_n-\boldsymbol{p}_{\mathrm{RX},m}|^3}(y_n-y_{\mathrm{RX},m}),\\
    \mathbf{G}_m({n,2}) &= -\frac{\partial f^d_m(n)}{\partial y_{\mathrm{RX},m}}=\frac{|\boldsymbol{v}_n\times(\boldsymbol{p}_n-\boldsymbol{p}_{\mathrm{RX},m})|}{\lambda|\boldsymbol{p}_n-\boldsymbol{p}_{\mathrm{RX},m}|^3}(x_n-x_{\mathrm{RX},m}).
    \end{aligned}
\end{equation}

Moreover, the second-order gradient (Hessian matrix) of $g_m(\mathbf{J}^{(\kappa-1)},\boldsymbol{p}_{\mathrm{RX},m},\widetilde{\mathbf{z}}_m))$ is given by
\begin{equation}
    \nabla^2 g_m(\mathbf{J}^{(\kappa-1)},\boldsymbol{p}_{\mathrm{RX},m},\widetilde{\mathbf{z}}_m)\!=\!2\mathbf{G}_m^{\mathsf{T}}\mathbf{G}_m\!+\!2\sum^{N-1}_i \boldsymbol{e}_{m}(i)\nabla^2\boldsymbol{e}_{m}(i)\label{con:b}.
\end{equation}

By the Gauss–Newton approximation, the second-order term in \cref{con:b} is neglected by assuming small residuals near the optimum. Thus, the Hessian matrix is approximated as
\begin{equation}
    \mathbf{B} \approx 2\mathbf{G}_m^{\mathsf{T}}\mathbf{G}_m.
\end{equation}

\bibliographystyle{IEEEtran}
\bibliography{Reference}

\vspace{12pt}

\end{document}